\documentclass[10pt,onecolumn,oneside]{article}

\usepackage{amsmath,amssymb,amsfonts,amsthm}
\usepackage{authblk}
\usepackage{stackengine}
\usepackage{longtable}
\usepackage{algorithm,algorithmicx,algpseudocode}
\usepackage{graphicx}
\usepackage{xcolor}
\usepackage{url,hyperref}

\usepackage[font=small,skip=6pt]{caption}

\allowdisplaybreaks 

\theoremstyle{plain} \newtheorem{lemma}{\textbf{Lemma}}
\theoremstyle{plain} 
\theoremstyle{remark} \newtheorem{remark}{\textbf{Remark}}
\theoremstyle{plain} \newtheorem{theorem}{\textbf{Theorem}}
\theoremstyle{plain} 
\theoremstyle{plain} 
\theoremstyle{definition} 
\theoremstyle{plain} 
\theoremstyle{plain}

\newtheoremstyle{thmstyleAAA3}{3pt}{3pt}{}{}{\bfseries}{.}{0.5em}{}
\theoremstyle{thmstyleAAA3} 
\newcommand{\rmv}[1]{}
\newcommand{\redtext}[1]{\color{black}{}\color{red}{#1}\color{black}{}}

\newcommand{\cyantext}[1]{\color{black}{}\color{cyan}{#1}\color{black}{}}

\newcommand{\greentext}[1]{\color{black}{}\color{green}{#1}\color{black}{}}
\newcommand{\boldgraytext}[1]{\color{black}{}\color{gray}\textbf{#1}\color{black}{}}

\makeatletter
\let\@@pmod\pmod
\DeclareRobustCommand{\pmod}{\@ifstar\@pmods\@@pmod}
\def\@pmods#1{\mkern4mu({\operator@font mod}\mkern 6mu#1)}
\makeatother

\makeatletter
\newenvironment{breakablealgorithm}
  {
   \begin{center}
     \refstepcounter{algorithm}
     \hrule height.8pt depth0pt \kern2pt
     \renewcommand{\caption}[2][\relax]{
       {\raggedright\textbf{\ALG@name~\thealgorithm} ##2\par}%
       \ifx\relax##1\relax 
         \addcontentsline{loa}{algorithm}{\protect\numberline{\thealgorithm}##2}%
       \else 
         \addcontentsline{loa}{algorithm}{\protect\numberline{\thealgorithm}##1}%
       \fi
       \kern2pt\hrule\kern2pt
     }
  }{
     \kern2pt\hrule\relax
   \end{center}
  }
\makeatother

\floatname{algorithm}{Procedure}

\renewcommand{\algorithmicreturn}[1]{\bgroup\\  ~#1\egroup}
\renewcommand{\algorithmiccomment}[1]{\bgroup\hfill//~#1\egroup}
\title{A generalization of the Von Neumann extractor}
\author{Claude Gravel}
\affil{\stackunder{\small{EAGLYS Inc.}}{\stackunder{\small{Tokyo, Japan}}{\stackunder{\small{\texttt{claudegravel1980@gmail.com}}}{\small{\texttt{c\_gravel@eaglys.co.jp}}}}}}
\date{\today}

\begin{document}




\maketitle

\begin{abstract}
An iterative randomness extraction algorithm which generalized the Von Neumann's extraction algorithm is detailed, analyzed and implemented in standard C++. Given a sequence of independently and identically distributed biased Bernoulli random variables, to extract randomness from the aforementioned sequence pertains to produce a new sequence of independently and identically distributed unbiased Bernoulli random variables. The iterative construction here is inspired from the work of Stout and Warren \cite{SouWar_1984} who modified appropriately the tree of probabilities produced by recursively repeating the Von Neumann's extraction algorithm. The correctness of the iterative algorithm is proven. The number of biased Bernoulli random variables needed to produce one unbiased instance is the complexity of interest. The complexity depends on the bias of the source. The expected complexity converges toward $3.10220648$\ldots~when the bias tends to $0$ and diverges when the bias tends to $1/2$. In addition to the expected complexity, some other results that concern the limiting asymptotic construction, and that seem unnoticed in the literature so far, are proven.

\end{abstract}

\section{Definition of the problem}

Given a binary alphabet $\mathcal{A}=\{0,1\}$, and a Bernoulli distribution on $\mathcal{A}$ defined by the probability vector $\mathbf{p}=(p_0,p_1)=(1-p,p)$ for some $0<p<1$, consider an infinite length binary random sequence $\mathbf{X}=(X_{i})_{i\in\mathbb{N}}\in\mathcal{A}^{\mathbb{N}}$. The random variables $X_i$ are independent of each other. Suppose that $p$ is \emph{unknown}, that $p$ cannot be determined exactly or that a statistical estimation is unacceptable like in cryptographic settings for instance. Mechanisms that produce sequences of independently and identically distributed \emph{biased} bits, abbreviated by i.i.d.~hereafter, with partially or unknown bias need to be de-biased such as in Grass et al.~\cite{Gra_ETAL_2020} for instance. Said differently, de-biasing a biased sequence is about extracting the randomness from the aforementioned biased sequence to produce a new unbiased and shorter sequence. Therefore, how can we extract i.i.d.~unbiased bits from the sequence $\mathbf{X}$ in a way that minimizes the number of consumed biased bits on average? To do that, we shall re-explore an idea of Von Neumann \cite{vN_1951}, and build upon it a strategy inspired from Stout and Warren \cite{SouWar_1984}. We are interesting here in a useful and efficient implementation no matter $p$. A non-exhaustive list of research articles discuss the generation of unbiased coins from i.i.d.~biased coins such as Bernard and Letac \cite{BerLet_1973}, Dwass \cite{Dwa_1972}, Elias \cite{Eli_1972}, Hoeffding and Gordon \cite{HoeGor_1970}, Pae and Loui \cite{PaeLou_2006}, Samuelson \cite{Sam_1968}, and Uehara \cite{Ryu_1995}. A comprehensive survey about uniform random generation is contained in L'\'{E}cuyer \cite{Ecu_2017}. The reverse problem of producing non-uniform discrete random variables from a sequence of unbiased i.i.d.~bits have been studied for instance in the last chapter of Devroye \cite{Dev_book1986}, the short survey from Gravel and Devroye \cite{DevGra_2020}, the work Han and Hoshi \cite{HanHos_1997}, and Knuth and Yao \cite{KnuYao_1976}.

In this article, the unbiased output are denoted by \texttt{T} (tail) and by \texttt{H} (head). The symbol $\mathbf{P}$ is generic and is used to denote the probability of an event with respect to its underlying probability space; the context shall render clear to which probability space we refer to. We use $\mathbf{E}$ to denote the expectation of a random variable. Upper-case letters denote random variables and lower-case letters denote their realizations.

We emphasize that $p$ is \emph{unknown} here. In the case of a known $p$, algorithms that fall under the Bernoulli factory umbrella have better performances. A detailed \emph{iterative} implementation is provided in section \ref{SECT_gen_opt_algo} with its correctness and efficiency proven. A C++ implementation can be found at \url{https://github.com/63EA13D5/}. As shown in Stout and Warren \cite{SouWar_1984}, there is no optimal algorithm for the extraction problem. When the bias tends to $0$ (or $p\to 1/2$), the expected complexity tends to $3.102206486$\ldots~biased bits for one unbiased bit as it will be shown in section \ref{SECT_comp_stat}. A method to find the expected complexity for general $p$ is established in section \ref{SECT_comp_stat}.

\section{A general extraction algorithm}\label{SECT_gen_opt_algo}

We recall briefly Von Neumann's idea that consists to split the sequence $\mathbf{X}$ into blocks of length $2$, and map a block with $01$ to $\mathtt{T}$, a block with $10$ to $\mathtt{H}$, and to discard any occurrence of $00$ or $11$. The procedure just described is amenable to a tree representation as on figure \ref{FIG_orig_algo}. The outputs are denoted by square leaves labelled by either \texttt{H} or \texttt{T}. Any discarded blocks yield to a repetition of the procedure shown by circular nodes labelled by \texttt{R} that we shall call restart nodes. We use the convention that an edge from a parent node to a left child represents a $0$ and a $1$ for the right child.

\begin{figure}[ht]
\begin{centering}
\includegraphics{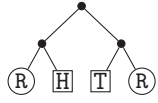}
\caption{Original von Neumann extractor using pair of consecutive bits}\label{FIG_orig_algo}
\end{centering}
\end{figure}

The correctness of the procedure follows from the fact that the events $\{X_1=0,X_2=1\}$ and $\{X_1=1,X_2=0\}$ are equally likely, that is, $\mathbf{P}\{X_1=0,X_2=1\}=pq=\mathbf{P}\{X_1=1,X_2=0\}$ where $q=1-p$. Repeating the procedure until success follows a geometric random process. Many random bits are discarded and are lost forever if we would simply repeat the original algorithm of Von Neumann. In an effort to maximize the use randomness, or entropy, contained in $\mathbf{X}$, blocks with different lengths is a natural strategy to build a code. For that, we may fix a maximal length for a codeword, say $m$ bits for a binary alphabet. Since in theory all codewords of lengths less than $m$ bits are admissible, then there is a maximum of $\sum_{i=0}^{m}{2^{m}}=2^{m+1}-1$ strings, vectors or codewords in our sampling code that we denote by $\mathcal{C}$. Every codeword $w\in\mathcal{C}$ is assigned the probability $p^{|w|}q^{|w|-\omega(w)}>0$, where $|w|$ denotes the length of $w$ and $\omega(w)$ is the number of non-zero elements (or equivalently said the Hamming weight). We do not challenge the completeness of probability spaces and necessarily we have as well $\sum_{w\in\mathcal{C}}{p^{|w|}q^{|w|-\omega(w)}}=1$. To be correct, we must have an encoding that partitions the codewords into three subsets that are identified with the likelihood to output a `head' (\texttt{H}), a `tail' (\texttt{T}) or to restart (\texttt{R}); therefore we have $\mathcal{C}=\mathcal{C}_{\texttt{H}}\cup\mathcal{C}_{\texttt{T}}\cup\mathcal{C}_{\texttt{R}}$ where the sets $\mathcal{C}_{\texttt{H}}$, $\mathcal{C}_{\texttt{T}}$ and $\mathcal{C}_{\texttt{R}}$ are disjoint and such that
\begin{align*}
\sum_{w\in\mathcal{C}_{\texttt{H}}}{p^{|w|}q^{|w|-\omega(w)}}&=\sum_{w\in\mathcal{C}_{\texttt{T}}}{p^{|w|}q^{|w|-\omega(w)}}\\
\sum_{w\in\mathcal{C}_{\texttt{R}}}{p^{|w|}q^{|w|-\omega(w)}}&=1-2\sum_{w\in\mathcal{C}_{\texttt{H}}}{p^{|w|}q^{|w|-\omega(w)}}.
\end{align*}

Let us see how to get a code with unequal lengths of blocks. If we repeat one more time the case shown on figure \ref{FIG_orig_algo}, then we have the situation represented by figure \ref{FIG_orig_algo_repeat_bis}.
\begin{figure}[ht]
\begin{centering}
\includegraphics{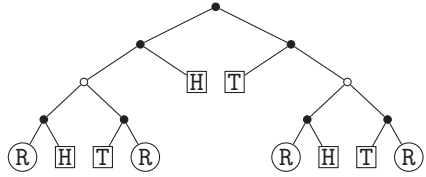}
\caption{One repetition of the extraction procedure shown on figure \ref{FIG_orig_algo}.}\label{FIG_orig_algo_repeat_bis}
\end{centering}
\end{figure}
The root of a repeated subtree is a small white circle. If we would repeat ad vitam aeternam, then leaves would have even depths only. Could we do better? The answer is yes, and as suggested before, by using blocks of different lengths. The four restart nodes on figure \ref{FIG_orig_algo_repeat_bis} could be replaced again, and nothing would be gained. We need to relabel some of the restart nodes while maintaining equal probability of the outcomes. We could remove simply the restart nodes in a symmetrical way, and this approach would leave the resulting tree with unary and binary nodes which is clearly not compressed. We replace the second and third restart node, reading from the left to the right, by \texttt{H} and \texttt{T}, respectively so that $\greentext{p^{2}qp}+\cyantext{p^{2}q^{2}}= \redtext{p^{2}q}$ and $\greentext{q^{2}pq}+\cyantext{q^{2}p^{2}}=\redtext{q^{2}p}$. In other words, we prune appropriately the tree represented on figure \ref{FIG_orig_algo_repeat_bis} and obtain the procedure represented on figure \ref{FIG_prune}.
\begin{figure}[ht]
\begin{centering}
\includegraphics{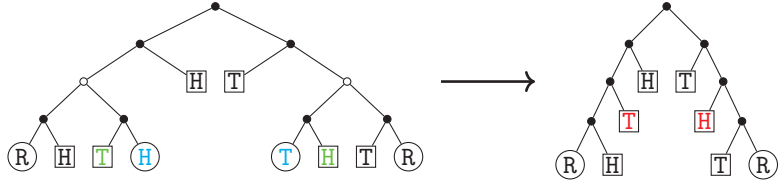}
\caption{Pruning of the tree from figure \ref{FIG_orig_algo_repeat_bis}. The transformation here allows to consume fewer bits on \emph{average} than on figure \ref{FIG_orig_algo_repeat_bis}.}\label{FIG_prune}
\end{centering}
\end{figure}

At this stage, we may wonder how many random bits from $\mathbf{X}$ are consumed on average whether we repeat the strategy based on figure \ref{FIG_orig_algo} or from the right side of figure \ref{FIG_prune}. For convenience, we denote by $\gamma_{\mathcal{C}}$ the equal likelihood of $\texttt{T}$ or $\texttt{H}$ for a given $\mathcal{C}$ which clearly depends on the \emph{unknown} $p$ and generally how we construct $\mathcal{C}$ as well, that is
\begin{displaymath}
\gamma_{\mathcal{C}}=\sum_{w\in\mathcal{C}_{\texttt{H}}}{p^{|w|}q^{|w|-\omega(w)}}=\sum_{w\in\mathcal{C}_{\texttt{T}}}{p^{|w|}q^{|w|-\omega(w)}}.
\end{displaymath}
Clearly if we restart $i>0$ times and succeed at the $(i+1)$-th time, then the expected number of bits consumed from $\mathbf{X}$ is $m\cdot i+H(\mathcal{C})$ where $m$ is the maximal length of a word from $\mathcal{C}$ and $H(\mathcal{C})$ is the average length of codewords. Equivalently, the average codeword length is the entropy of the probability distribution over $\mathcal{C}$. Let $N_{\mathcal{C}}$ be the random number of biased bits that are consumed. Then $N_{\mathcal{C}}$ is a geometric variable and we have
\begin{align}
\mathbf{E}(N_{\mathcal{C}})&=\sum_{i=0}^{\infty}{\big(m\cdot{i}+H(\mathcal{C})\big)(1-2\gamma_{\mathcal{C}})^{i}(2\gamma_{\mathcal{C}})}=\frac{m(1-2\gamma_{\mathcal{C}})}{2\gamma_{\mathcal{C}}}+H(\mathcal{C}).\label{expect_geom}
\end{align}
Asymptotically if $\mathcal{C}$ is designed to contain words of arbitrary lengths, that is $m$ is not bounded, then we must have that $\gamma_{C}\to \frac{1}{2}$ faster than $m\to\infty$ so that $\mathbf{E}(N_{\mathcal{C}})\to H(\mathcal{C})$. We will come back to the analysis of the expected complexity in section \ref{SECT_comp_stat}, and, more precisely, the analysis of $H(\mathcal{C})$. Sampling codes that are efficient necessarily minimizes $H(\mathcal{C})$.

It might be pedagogical to repeat the construction one more time by using the right tree on figure \ref{FIG_prune}. If we use the latter tree to represent our sampling code to build a new tree of height $8$ and pruning accordingly, then we have that $\greentext{p^{4}q^{4}}+\cyantext{p^{4}q^{3}p}= \redtext{p^{4}q^{3}}$ and $\greentext{q^{4}p^{4}}+\cyantext{q^{4}p^{3}q}=\redtext{q^{4}p^{3}}$. The resulting tree is the one displayed on the right of figure \ref{FIG_double_rep_prune}.
\begin{figure}[ht]
\begin{centering}
\includegraphics{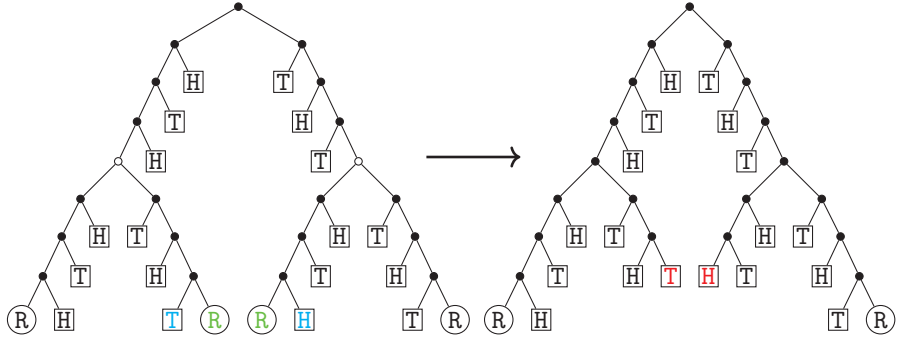}
\caption{Repeating the structure on the right side of figure \ref{FIG_prune} and pruning accordingly.}\label{FIG_double_rep_prune}
\end{centering}
\end{figure}

We make now a few observations about binary tree-based extraction algorithms. By the notation $\overline{w}$, we mean the word obtained from $w$ by flipping all of its bits.
\begin{remark}\label{rem_necessary_cond}
If a binary tree-based extraction algorithm satisfies the following conditions, then it is correct.
\begin{enumerate}
\item[1.] We have that $w\in\mathcal{C}_{\texttt{H}}\Leftrightarrow \overline{w}\in\mathcal{C}_{\texttt{T}}$, and $w\in\mathcal{C}_{\texttt{R}}\Leftrightarrow\overline{w}\in\mathcal{C}_{\texttt{R}}$. We can think of the previous equivalences as a type of symmetry.
\item[2.] Leaves must be labelled in an alternating way whenever walking along the leaves.
\end{enumerate}
The following conditions seem necessary for the algorithm to be efficient. We recall that Stout and Warren \cite{SouWar_1984} showed that no optimal algorithm exists.
\begin{enumerate}
\item[1.] The entropy $H(\mathcal{C})$ must be as small as possible. Equivalently, the expected height of the underlying tree must be as small as possible.
\item[2.] The $\#\mathcal{C}_{\texttt{R}}$ must be as small as possible, that is $2$.
\end{enumerate}
\end{remark}

Before detailing the general procedure, we introduce two symbols: $\mathcal{T}$ for the asymptotic tree and $\mathcal{T}_{k}$ for the tree of height $2^{k}$ obtained by trimming $\mathcal{T}$. The trees on figure \ref{FIG_orig_algo}, the right of figure \ref{FIG_prune}, and the right of figure \ref{FIG_double_rep_prune} represent therefore $\mathcal{T}_{1}$, $\mathcal{T}_{2}$ and $\mathcal{T}_{3}$, respectively. We now detail a procedure to construct $\mathcal{T}_{k}$ for arbitrary values of $k\geq 1$. More specifically, we obtain iteratively the codewords that corresponds to the binary representations of the leaves from $\mathcal{T}_{k}$. Before exhibiting the iterative method, we analyze the recursive nature of the problem as hopefully suggested from the previous figures.

Suppose that we know how to generate $\mathcal{T}_{k-1}$. Then let us use the knowledge of $\mathcal{T}_{k-1}$ to build $\mathcal{T}_{k}$ in the way as $\mathcal{T}_{2}$ is built from $\mathcal{T}_{1}$ on figure \ref{FIG_prune}, and as $\mathcal{T}_{3}$ is built from $\mathcal{T}_{2}$ on figure \ref{FIG_double_rep_prune}. The probability to output a symbol \texttt{H}, or equally likely a \texttt{T}, is denoted by $\gamma_{k}$ from now on. The resulting probability satisfy the following recurrence:
\begin{align}
\gamma_{k}&=(1+p^{2^{k-1}}+q^{2^{k-1}})\gamma_{k-1}+p^{2^{k-1}}q^{2^{k-1}}\quad\text{for $k\geq 2$},\label{rec_expr}\\
\gamma_{1}&=pq\nonumber
\end{align}
Expression (\ref{rec_expr}) must be symmetrical as a bivariate polynomial function. A different way to observe the symmetry is to use the bivariate generating function as in Flojolet and Sedgewick \cite{FlaSed_BOOK_2009} or Szpankowski \cite{Szp_BOOK_2001}. In essence, the generalization of the Von Neumann extractor proposed here hides a ternary structure. Indeed, the operational meaning of expression (\ref{rec_expr}) is as follow:
\begin{enumerate}
\item[1.] Consider the trimmed tree $\mathcal{T}_{k-1}$ of height $2^{k-1}$ with leftmost branch having probability $p^{2^{k-1}}$ and with rightmost branch having probability $q^{2^{k-1}}$. The leftmost and rightmost branches yield to discard the blocks $0^{2^{k-1}}$ and $1^{2^{k-1}}$, respectively.
\item[2.] Create $\mathcal{T}_{k}$ from $\mathcal{T}_{k-1}$ by linking to the latter two more copies of itself, one at the leftmost branch and one at rightmost branch. In this way, we observe that $\mathcal{T}_{k}$ has three copies of $\mathcal{T}_{k-1}$ which explain the term $(1+p^{2^{k-1}}+q^{2^{k-1}})$ in the expression (\ref{rec_expr}). As a result, we obtain indeed a symmetrical construction, but the tree has four restart nodes among which the leftmost and the rightmost are kept. What about the two branches in the middle? The answer is item $3$ right now.
\item[3.] To deal with the branches represented by the strings $0^{2^{k-1}}1^{2^{k-1}-1}0=s_{l}$ and $1^{2^{k-1}}0^{2^{k-1}-1}1=s_{r}$, we need to transform the tree while keeping the label consistent, that is alternating, and while maintaining symmetry. We therefore replace those two middle restart nodes with two nodes with branches represented by $0^{2^{k-1}}1^{2^{k-1}}=s'_{l}$ and $1^{2^{k-1}}0^{2^{k-1}}=s'_{r}$. Each new nodes contribute equally likely to the probabilities of outputting \texttt{H} or \texttt{T}, and that explains the $p^{2^{k-1}}q^{2^{k-1}}$ term in the expression (\ref{rec_expr}).
\end{enumerate}

We give here an iterative construction of the trimmed tree $\mathcal{T}_{k}$ where $k$ is the binary logarithm of the height. In order to proceed iteratively, let us expand now the expression (\ref{rec_expr}). Then we have
\begin{align}
\gamma_{k}&=\left(1+p^{2^{k-1}}+q^{2^{k-1}}\right)\gamma_{k-1}+p^{2^{k-1}}q^{2^{k-1}}\nonumber\\
&=p^{2^{k-1}}q^{2^{k-1}}+\sum_{i=0}^{k-2}{p^{2^{i}}q^{2^{i}}\prod_{j=i+1}^{k-1}{\left(1+p^{2^{j}}+q^{2^{j}}\right)}}\label{count_nb_terms_expr}\\
&=p^{2^{k-1}}q^{2^{k-1}}+\sum_{i=0}^{k-2}{p^{2^{i}}q^{2^{i}}\left(\sum_{\ell=0}^{3^{k-1-i}-1}{\left(\prod_{m=0}^{k-2-i}{\alpha_{\ell_{m}}^{2^{i+1+m}}}\right)}\right)}\label{bivar_poly_proto},
\end{align}
where $\alpha_{0}=1$, $\alpha_{1}=p$, $\alpha_{2}=q$, and $\ell_{m}$ is the $m$-th coefficient of the ternary expansion of $\ell$ such that $0\leq\ell < 3^{k-1-i}$.
\begin{lemma}\label{lem_nb_leaves}
There are exactly $3^{k-1}$ terms in the expansion of $\gamma_{k}$, and therefore there are $3^{k-1}$ leaves which are labelled by $\texttt{T}$, and similarly for $\texttt{H}$.
\end{lemma}
\begin{proof}
We proceed by induction. As the base case, we consider the first two terms from the summation in (\ref{count_nb_terms_expr}), that is, the terms corresponding to $i=0$ and $i=1$. We have that
\begin{align}
&pq\prod_{j=1}^{k-1}{\left(1+p^{2^{j}}+q^{2^{j}}\right)}+p^{2}q^{2}\prod_{j=2}^{k-1}{\left(1+p^{2^{j}}+q^{2^{j}}\right)}\nonumber\\
&=\left(\prod_{j=2}^{k-1}{\left(1+p^{2^{j}}+q^{2^{j}}\right)}\right)\left(pq(1+p^{2}+q^{2})+p^{2}q^{2}\right)\label{onara1_eq}\\
&=\left(\prod_{j=2}^{k-1}{\left(1+p^{2^{j}}+q^{2^{j}}\right)}\right)\left(pq+p^{2}q+pq^{3}\right)\label{onara2_eq}
\end{align}
The number of terms from expression (\ref{onara1_eq}) to (\ref{onara2_eq}) decreases from $3^{k-2}+3^{k-1}=3^{k-2}(3+1)$ to $3^{k-2}(3)=3^{k-1}$, respectively. For convenience, let us denote $\tau_{2}=pq+p^{2}q+pq^{3}$ so that expression (\ref{count_nb_terms_expr}) can be rewritten as
\begin{align*}
&p^{2^{k-1}}q^{2^{k-1}}+\left(\prod_{j=2}^{k-1}{\left(1+p^{2^{j}}+q^{2^{j}}\right)}\right)\tau_{2}+\sum_{i=1}^{k-2}{p^{2^{i}}q^{2^{i}}\prod_{j=i+1}^{k-1}{\left(1+p^{2^{j}}+q^{2^{j}}\right)}}
\end{align*}
For the induction step, suppose that the first $i'<k-1$ partial sums from the summation in expression (\ref{count_nb_terms_expr}) evaluate to
\begin{align}
&p^{2^{k-1}}q^{2^{k-1}}+\left(\prod_{j=i'}^{k-1}{\left(1+p^{2^{j}}+q^{2^{j}}\right)}\right)\tau_{i'}+\sum_{i=i'}^{k-2}{p^{2^{i}}q^{2^{i}}\prod_{j=i+1}^{k-1}{\left(1+p^{2^{j}}+q^{2^{j}}\right)}},\label{onara3_eq}
\end{align}
for some probability $\tau_{i'}$. From (\ref{onara3_eq}), we group the term $p^{2^{i'}-1}q^{2^{i'}+1}$ in $\tau_{i'}$ with the factor $p^{2^{i'}}q^{2^{i'}}$ in the summation. Then we obtain $p^{2^{i'}-1}q^{2^{i'}+1}+p^{2^{i'}}q^{2^{i'}}=p^{2^{i'}-1}q^{2^{i'}}$, and the number of terms decreases from $3^{k-1-i'}(3^{i'}+1)$ to $3^{k-1-i'}(3^{i'})=3^{k-1}$.
\end{proof}

\begin{remark}
The quantity $\tau_{i}$ in the proof of lemma \ref{lem_nb_leaves} is the halting probability with either a tail or head symbol for the algorithm represented by $\mathcal{T}_{i}$.
\end{remark}

On input $k>1$, algorithm \ref{ALG_get_strings_in_language} outputs all the strings representing the leaves in $\mathcal{T}_{k}$. In the algorithms, the generic pseudo-coding style operator $\pi$ accesses the coordinate of either a tuple, a string object or, in general, any ordered object. For instance $\pi_{1}(a,b)=a$ for a pair and $\pi_{3}(010)=0$ for a string. Also $|x|$ means the length or the size of an object and it should be clear from the context what type of object.
\begin{breakablealgorithm}\label{ALG_get_strings_in_language}
\caption{List representation of $\mathcal{T}_{k}$}
\begin{algorithmic}[1]
\raggedright
\Require An integer $k>1$ for the binary logarithm of the height for the trimmed tree.
\Ensure List $L$ of size $2\cdot 3^{k}$. $L_{i}$ is a pair for $0\leq i < 2\cdot3^{k-1}$. A pair is made of a string in its first coordinate and a boolean its second coordinate.
\State{Insert $(01,\texttt{H})$ to $L$}
\State{Insert $(10,\texttt{T})$ to $L$}
\For{$i=1$ \textbf{to} $k-1$}\Comment{That is $1\leq i < k$.}
\State{$A_{0}\leftarrow 0^{2^{i}}$}
\State{$A_{1}\leftarrow 1^{2^{i}}$}
\State{$L^{\star}\leftarrow L$}
\For{$\ell = 0$ \textbf{to} $|L^{\star}|-1$}\Comment{That is $0\leq \ell < |L^{\star}|$.}
\State{$S_{0}\leftarrow A_{0}\|\pi_{1}(L^{\star}_{i})$}\label{ALG_get_str_line_cat1}
\State{$S_{1}\leftarrow A_{1}\|\pi_{1}(L^{\star}_{i})$}\label{ALG_get_str_line_cat2}
\If{$|\pi_{1}(L^{\star}_{\ell})| = 2^{i}$ \textbf{and} $\pi_{1}(\pi_{1}(L^{\star}))=0$}
\State{Remove the last character from $S_{1}$.}\label{ALG_rem_char1}
\EndIf
\If{$|\pi_{1}(L^{\star}_{\ell})| = 2^{i}$ \textbf{and} $\pi_{1}(\pi_{1}(L^{\star}))=1$}
\State{Remove the last character from $S_{0}$.}\label{ALG_rem_char2}
\EndIf
\State{Insert $(S_0,\pi_{2}(L^{\star}))$ to $L$}
\State{Insert $(S_1,\pi_{2}(L^{\star}))$ to $L$}
\EndFor
\EndFor
\end{algorithmic}
\end{breakablealgorithm}

We introduce some notation. For a set $S$ (or a list) of strings and a string $y$, then $yS$ denotes the set $\{y\|z\colon z\in S\}$. If $S$ is a set of pairs $(z,b)$ where $z$ is a string and $b$ is a single bit, then, for convenience, we write $yS=\{(yz,b)\colon (y,b)\in S\}$.

\begin{theorem}
Algorithm \ref{ALG_get_strings_in_language} outputs iteratively the list representation of $\mathcal{T}_{k}$.
\end{theorem}
\begin{proof}
The proof is by induction. Initially for the base step, $L_{1}$ and $L_{2}$ contain the pairs $(01,\texttt{H})$ and $(10,\texttt{T})$ that encode $\mathcal{T}_{2^{1}}$. For the inductive step, suppose the cells $L_{i}$ for $1\leq i\leq 2\cdot 3^{j-1}$ are pairs that encode $\mathcal{T}_{j}$ for $j<k$. Because the tree $\mathcal{T}_{k}$ is made by appending $\mathcal{T}_{k-1}$ to itself twice, one copy hanging on the leftmost branch and another copy hanging on the rightmost branch followed by the pruning step. We recall (1) and (2) from remark \ref{rem_necessary_cond} and emphasize that appending keeps alternating the order of heads and tails and keeps equal the number of heads and tails at every level. From a data structure and operational perspective, the appending step is equivalent to compute $0^{2^{k-1}}L_{k-1}\cup 1^{2^{k-1}}L_{k-1}$ which is achieved by lines (\ref{ALG_get_str_line_cat1}) and (\ref{ALG_get_str_line_cat2}). The list $0^{2^{k-1}}L_{k-1}\cup 1^{2^{k-1}}L_{k-1}$ contains only two elements of length $2^{k}$, and those strings are exactly those leaves that are pruned. In terms of strings, the pruning simply corresponds to removing the last character which is achieved by lines (\ref{ALG_rem_char1}) and (\ref{ALG_rem_char2}).
\end{proof}

Clearly it holds that $\gamma_{k-1}<\gamma_{k}$ for all $k\geq 1$ and that $2\gamma_{k}<1$. Therefore we have that $\gamma_{k}$ converges as $k$ tends to infinity since it forms a bounded sequence of increasing terms.

\begin{lemma}
We have that
\begin{displaymath}
\lim_{k\to\infty}{\gamma_{k}}=\frac{1}{2}.
\end{displaymath}
In other words, the asymptotic tree $\mathcal{T}$ covers the whole interval $(0,1)$.
\end{lemma}
\begin{proof}
The proof is by induction. Given the paragraph that precedes the lemma, we need only to show that $2\gamma_{k}+p^{2^{k}}+q^{2^{k}}=1$ for all $k\geq 1$ since the latter two terms tend to zero as $k$ tends to infinity. For the base case with $k=1$,
we observe that
\begin{displaymath}
2\gamma_{1}+p^{2^{1}}+q^{2^{1}}=2pq+p^{2}+q^{2}=(p+q)^{2}=1.
\end{displaymath}
Let us assume for the inductive step that $2\gamma_{j}+p^{2^{j}}+q^{2^{j}}=1$ for $1\leq j<k$. Then we have
\begin{align*}
2\gamma_{k}&=\left(\left(1+p^{2^{k-1}}+q^{2^{k-1}}\right)\gamma_{k-1}+p^{2^{k-1}}q^{2^{k-1}}\right)+\\
&\qquad \left(\left(1+p^{2^{k-1}}+ q^{2^{k-1}}\right)\gamma_{k-1}+p^{2^{k-1}}q^{2^{k-1}}\right)\\
&=2\gamma_{k-1}+2\gamma_{k-1}\left(p^{2^{k-1}}+q^{2^{k-1}}\right)+2p^{2^{k-1}}q^{2^{k-1}}.
\end{align*}
By adding $p^{2^{k}}+q^{2^{k}}$ to the previous equality, we then obtain
\begin{align}
2\gamma_{k}+p^{2^{k}}+q^{2^{k}}&=2\gamma_{k-1}+\left(p^{2^{k-1}}+q^{2^{k-1}}\right)\left(2\gamma_{k-1}\right)+\nonumber\\
&\qquad p^{2^{k-1}}\left(p^{2^{k-1}}+q^{2^{k-1}}\right)+q^{2^{k-1}}\left(p^{2^{k-1}}+q^{2^{k-1}}\right)\nonumber\\
&=2\gamma_{k-1}+\nonumber\\
&\qquad p^{2^{k-1}}\left(2\gamma_{k-1}+p^{2^{k-1}}+q^{2^{k-1}}\right)+\label{eq_boulette1}\\
&\qquad q^{2^{k-1}}\left(2\gamma_{k-1}+p^{2^{k-1}}+q^{2^{k-1}}\right).\label{eq_boulette2}
\end{align}
The induction hypothesis implies that terms within parentheses of (\ref{eq_boulette1}) and (\ref{eq_boulette2}) are both equal to $1$. Then we use the induction hypothesis one more time to complete the proof.
\end{proof}

By sorting the outputs of algorithm \ref{ALG_get_strings_in_language} into ascending order of lengths, the extraction process can be sped up by using an array with random access to lists containing strings of the same lengths so that to yield an iterative algorithm. Lemma \ref{LEM_no_str_cong_1_mod4} is therefore useful to speed up the extraction given strings in the language defined by the extraction problem.

\begin{lemma}\label{LEM_no_str_cong_1_mod4}
There is no string $x\in\{0,1\}^{\mathbb{N}}$ that encodes a leaf in $\mathcal{T}$ such that $|x|\equiv 1\bmod{4}$.
\end{lemma}
\begin{proof}
The proof is yet by induction. This time, we recall bivariate expression (\ref{bivar_poly_proto}). There are integers $c_{ij}\geq 0$ such that $c_{ij}=0$ whenever $i+j>2^{k}$ for which
\begin{displaymath}
\gamma_{k}=\sum_{i=0}^{3^{k}-1}\sum_{j=0}^{3^{k}-1}{c_{ij}p^{i}q^{j}}.
\end{displaymath}
We need to prove that there is no pair of indices $(i,j)$ such that $i+j\equiv 1\bmod{4}$.

For the base case, we start at $k=2$. Since $\gamma_{2}=p^{3}q+pq+q^{2}p$, then clearly $3+1$, $1+1$, and $2+1$ are congruent to $0,2,3$ modulo $4$. For the inductive step, suppose that $\gamma_{\ell}$ contains no term $p^{i}q^{j}$ such $i+j\equiv 1\bmod{4}$ for $1< \ell\leq k-1$. Then both $p^{2^{k-1}}p^{i}q^{j}$ and $q^{2^{k-1}}p^{i}q^{j}$ yields $2^{k-1}+i+j\equiv i+j\not\equiv 1\bmod{4}$. We observe that the pruning step affects only the two leaves having depth $2^{k}$ which are replaced by one leaf with depth $2^{k}-1$. Since $2^{k}-1\equiv 3\bmod{4}$, then we are fine.
\end{proof}

Now suppose that we have an array $A$ indexed by $1\leq i\leq 2^{k}$. Each index is the length of strings representing leaves in $\mathcal{T}_{k}$. Therefore $A_{i}$ points to the list of strings of length $i$. By lemma \ref{LEM_no_str_cong_1_mod4}, we skip $A_{i}$ with $i\equiv 1\bmod{4}$. Algorithm \ref{ALG_extratorchon} also outputs the depth at which the leaf is located in $\mathcal{T}_{k}$ for this shall be useful later in section \ref{SECT_comp_stat}. For a fixed $k$, which we recall is the height of the binary logarithms of the trimmed asymptotic tree $\mathcal{T}$, the following extraction procedure may yield an output shorter than expected, and possibly empty, if the input string is too biased and $k$ too small.

\begin{breakablealgorithm}\label{ALG_extratorchon}
\caption{An iterative generalization of the Von Neumann extractor}
\begin{algorithmic}[1]
\raggedright
\Require A string $\mathbf{X}$ made of i.i.d.~biased random bits.
\Require An integer $k>1$ for the binary logarithm of the height for the trimmed tree.
\Require An array $A$ of list of strings as described above. $A_{i}$ denotes the $i$-th list. $A_{i,j}$ denotes the $j$-th element of the $i$-th list.
\Ensure A list $L$ of pairs of type $(\text{boolean},\text{positive integer})$. The first coordinate is an unbiased bit and the second coordinate the depth at which the corresponding leaf is located in $\mathcal{T}_{k}$.
\State{$i\leftarrow 0$}\Comment{Number of characters read from $\mathbf{X}$}
\State{$r\leftarrow 0$}\Comment{The class representative modulo $4$}
\State{$y\leftarrow \varepsilon$}\Comment{Substring of $\mathbf{X}$ that grows in size until we know which of the $A_{i}$'s it belongs.}
\State{$L\leftarrow \emptyset$}
\Repeat
\State{\verb!BEGINLOOP!}\Comment{Label}
\If{$i\leq |\mathbf{X}|$}
\State{\textbf{Break} out the loop}
\EndIf
\State{$c_0\leftarrow x_{i}$}
\State{$c_1\leftarrow x_{i+1}$}\Comment{Read two characters in virtue of lemma \ref{LEM_no_str_cong_1_mod4}}
\For{$j=0$ \textbf{to} $|A_{2+4r}|$}
\If{$y\|c_0\|c_1 = \pi_{1}(A_{2+4r,j})$}\Comment{Is $y\|c_0\|c_1\in A_{2+4r}$?}
\State{Insert $(\pi_{2}(A_{2+4r,j}),2+4r)$ to $L$.}
\State{$y\leftarrow \varepsilon$}
\State{$i\leftarrow i+2$}
\State{$r\leftarrow 0$}
\State{\textbf{Goto} \verb!BEGINLOOP!}
\EndIf
\EndFor
\State{$y\leftarrow y\|c_{0}\|c_{1}$}
\State{$i\leftarrow i+2$}
\State{$c_2\leftarrow x_{i+2}$}\Comment{Read one bit}
\For{$j=0$ \textbf{to} $|A_{3+4r}|$}
\If{$y\|c_2 = \pi_{1}(A_{3+4r,j})$}\Comment{Is $y\|c_0\|c_1\|c_2\in A_{3+4r}$?}
\State{Insert $(\pi_{2}(A_{3+4r,j}),3+4r)$ to $L$.}
\State{$y\leftarrow \varepsilon$}
\State{$i\leftarrow i+1$}
\State{$r\leftarrow 0$}
\State{\textbf{Goto} \verb!BEGINLOOP!}
\EndIf
\EndFor
\State{$y\leftarrow y\|c_{2}$}
\State{$i\leftarrow i+1$}
\State{$c_3\leftarrow x_{i+3}$}\Comment{Read one bit}
\For{$j=0$ \textbf{to} $|A_{4+4r}|$}
\If{$y\|c_3 = \pi_{1}(A_{4+4r,j})$}\Comment{Is $y\|c_0\|c_1\|c_2\|c_{3}\in A_{4+4r}$?}
\State{Insert $(\pi_{2}(A_{4+4r,j}),4+4r)$ to $L$.}
\State{$y\leftarrow \varepsilon$}
\State{$i\leftarrow i+1$}
\State{$r\leftarrow 0$}
\State{\textbf{Goto} \verb!BEGINLOOP!}
\EndIf
\EndFor
\State{$y\leftarrow y\|c_{3}$}
\State{$i\leftarrow i+1$}
\State{$r\leftarrow r+1$}\Comment{Once here, increment $r$. Also we have $|y|\equiv 0\bmod{4}$.}
\Until{$4r+3<|A|$}\Comment{If we use $4r+4$, then we may throw an out-of-bound exception.}
\end{algorithmic}
\end{breakablealgorithm}

Hopefully the reader is convinced at this stage of the correctness of algorithm \ref{ALG_extratorchon}. The probability that $4r+3\geq |A|$ is $p^{2^{k}}+q^{2^{k}}$. From a pragmatic point of view, for not too big $k$, algorithm \ref{ALG_extratorchon} almost never hits a subsequence of $\mathbf{X}$ that is not contained in any of the $A_{i}$'s, and therefore yield an output of the expected length.

To end this section, table \ref{TAB_timing} mentions the time of extraction for input sequences of length $2^{26}$ bits with different bias. The mean depth column is the average number of biased bits needed to produce one unbiased bit. The mean depth is therefore the sum over all depths divided by the output length of the extracted string. The binary logarithm of the height of the trimmed tree was set to $k=10$. It is almost impossible with the values of $p$ shown in table \ref{TAB_timing} to obtain $2^{10}$ consecutive zeros or $2^{10}$ consecutive ones that would force algorithm \ref{ALG_extratorchon} to output shorter list than expected. The standard library of the C++ programming language is used to implement the previous algorithms that can be found at \url{https://github.com/63EA13D5/}.

\begin{centering}
\renewcommand*{\arraystretch}{1.25}
\begin{longtable}{|c|c||c|c|c|}\caption{Extraction time in milliseconds and mean depth with $k=10$}\label{TAB_timing}\\
\hline
$p$ & Input length & Output length & Time & Mean depth\\
\hline
\endfirsthead
\hline
$p$ & Input length & Output length & Time & Mean depth\\
\hline
\endhead
\multicolumn{5}{r}{\boldgraytext{Continued on next page}}
\endfoot
\endlastfoot
\hline
$0.51$ & $2^{26}$ & $21624098$ & $3227.5877$ & $3.103429$ \\
\hline
$0.525$ & $2^{26}$ & $21557549$ & $2756.1444$ & $3.113010$ \\
\hline
$1-e^{-1}$ & $2^{26}$ & $19634433$ & $2553.9419$ & $3.417917$\\
\hline
$1/\sqrt{2}$ & $2^{26}$  & $16916921$ & $3698.6932$ & $3.966967$\\
\hline
$0.96875$ & $2^{26}$ & $2079646$ & $20926.1112$ & $32.269355$\\
\hline
\end{longtable}
\end{centering}

The mean depth is an empirical estimation of the expected complexity that we study next. With respect to the entries from table \ref{TAB_timing}, the mean depth is $2^{-26}\sum_{i=1}^{2^{26}}{Y_{i}}$ where $Y_i$ is the number of biased bits need to produce $B_i$. The mean depth is comparable with the ratio of the input length divided by the output length.

\section{Expected complexity}\label{SECT_comp_stat}

The main question of this section is how many input bits from $\mathbf{X}$ does algorithm \ref{ALG_extratorchon} need in order to extract one unbiased bit $B$, or more precisely the first coordinate of a pair $(B,Y)\in L$? The question is therefore what is $\mathbf{E}(Y)$? Clearly as $p\to 1$, we should expect that $Y\to\infty$ so must be $\mathbf{E}(Y)$.

Suppose for a while that \ref{ALG_extratorchon} runs with $k=\infty$ on some biased input $\mathbf{X}$ and that we stop as soon as one bit $B$ together with $Y$ is obtained. Then what is $\mathbf{E}(Y)$? We study the asymptotic quantity $\mathbf{E}(Y)$ through the sequence $(Y_1,Y_2,\ldots, Y_k,\ldots)$ where $\mathbf{E}(Y_{k})$ is the expected height of $\mathcal{T}_{k}$ which is the expected number of bits from $\mathbf{X}$ consumed by algorithm \ref{ALG_extratorchon} for a finite $k$. By the Lebesgue's dominated convergence theorem and fixed $p$, the quantity $\mathbf{E}(Y_{k})$ converges.

We recall the expression for $\gamma_{k}$ which is
\begin{align}
\gamma_{k}&=p^{2^{k-1}}q^{2^{k-1}}+\sum_{i=0}^{k-2}{p^{2^{i}}q^{2^{i}}\left(\sum_{\ell=0}^{3^{k-1-i}-1}{\left(\prod_{m=0}^{k-2-i}{\alpha_{\ell_{m}}^{2^{i+1+m}}}\right)}\right)}\nonumber\\
&=\sum_{i=0}^{3^{k}-1}\sum_{j=0}^{3^{k}-1}{c_{ij}p^{i}q^{j}},\label{poly_rep_H}
\end{align}
such that $c_{ij}=0$ whenever $i+j>2^{k}$, and where $\alpha_{0}=1$, $\alpha_{1}=p$ and $\alpha_{2}=q$. By definition, we have
\begin{displaymath}
\mathbf{E}(Y_{k})=\sum_{y=0}^{2^{k}}{y\mathbf{P}\{Y_k=y\}}\quad\text{where}\quad\mathbf{P}\{Y_k=y\}=\sum_{\substack{(i,j)\in\mathbb{N}^{2}\\i+j=y}}{c_{ij}p^{i}q^{j}}.
\end{displaymath}
The coefficients $\ell_{m}$ from the binary expansion of $\ell=(\ell_{0},\ldots,\ell_{k-2-i})$ such that $\ell_{m}=0$ do not contribute to $Y_{k}$ because $\alpha_{0}=1$.

We could use the output from algorithm \ref{ALG_get_strings_in_language}, store each strings with respect to their lengths, and compute probabilities using the Hamming weight (the count of the number of non-zero elements). We proceed slightly differently by computing directly all pairs $(i,j)$ in the bivariate polynomial expression of $\gamma_{k}$ given above.
\begin{breakablealgorithm}\label{ALG_get_prob_depth}
\caption{Computing pairs $(i,j)$ for the bivariate expansion of $\gamma_{k}$}
\begin{algorithmic}[1]
\raggedright
\Require An integer $k>1$.
\Ensure A list (or vector) $L$ of $k$ lists $L_{t}$ for $1\leq t\leq k$ such that $L_{t}$ contains all pairs $(i,j)$ appearing in $\gamma_{t}$.
\State{$L\leftarrow \emptyset$}
\State{Insert $(1,1)$ to $L_{1}$}\Comment{The only pair in $\mathcal{T}_{1}$ which has height $2^{1}$.}
\State{Insert $L_{1}$ to $L$}
\For{$t=2$ \textbf{to} $k$}
\State{$L'\leftarrow L_{t-1}$}\Comment{All pairs appearing in the expression of $\gamma_{t-1}$ also appears in $\gamma_{t}$.}
\For{$s=0$ \textbf{to} $|L'|$}
\State{$i_1\leftarrow \pi_{1}(L'_s)+2^{t-1}$}
\State{$j_1\leftarrow \pi_{2}(L'_s)$}\label{just_before_pruning_left}
\If{ $i_1=2^{t-1}-1$ \textbf{and} $j_1=2^{t-1}+1$}\Comment{This corresponds to pruning the left subtree.}\label{just_before_pruning_right}
\State{$j_1\leftarrow j_1-1$}
\EndIf
\State{$i_2\leftarrow \pi_{1}(L'_s)$}
\State{$j_2\leftarrow \pi_{2}(L'_s)+2^{t-1}$}
\If{ $i_2=2^{t-1}-1$ \textbf{and} $j_2=2^{t-1}+1$}\Comment{This corresponds to pruning the right subtree.}\label{if_statement_left}
\State{$j_2\leftarrow j_2-1$}
\EndIf
\State{Insert $(i_1,j_1)$ to $L'_s$}
\State{Insert $(i_2,j_2)$ to $L'_s$}
\EndFor
\State{Insert $L'$ to $L$}\Comment{Once here, $L'=L_{t}$.}
\EndFor
\end{algorithmic}
\end{breakablealgorithm}

\begin{proof}
The correctness of algorithm \ref{ALG_get_prob_depth} follows directly from the correctness of algorithm \ref{ALG_get_strings_in_language}.

We emphasize that there are two pairs $(i,j)$ such that $i+j=2^{t}$ for the left subtree once the execution of line (\ref{just_before_pruning_left}) completed, and those pairs are $(2^{t},0)$ (for the branch encoded by $0^{2^{t}}$) and the pair $(2^{t-1}-1,2^{t-1}+1)$ (for the branch encoded by $0^{2^{t-1}}1^{2^{t-1}-1}0$ ). The pruning affects the branch encoded by $0^{2^{t-1}}1^{2^{t-1}-1}0$ which explains the conditional if-statement at line (\ref{if_statement_left}). The same remark applies for the right subtree.
\end{proof}

Given input $k > 1$ to algorithm \ref{ALG_get_prob_depth} with output $L=(L_1,\ldots,L_{t},\ldots,L_{k})$ where $L_{t}$ is the list for all pairs $(i,j)$ in the expansion of $\gamma_{t}$ for $1\leq t\leq k$. We simply therefore compute the distribution of $Y_{k}$ as follow:
\begin{displaymath}
\mathbf{P}\{Y_{k}=y\}=\sum_{\substack{(i,j)\in L_{k}\\i+j=y}}{(p^{i}q^{j}+q^{i}p^{j})}.
\end{displaymath}
We observe algorithm (\ref{ALG_get_prob_depth}) outputs all $L_{t}$ such that $1\leq t\leq k$ from which we can compute $\mathbf{E}(Y)=\lim_{k\to\infty}\mathbf{E}(Y_{k})$ accurately as shown from the tables \ref{TAB_I}, \ref{TAB_II}, and \ref{TAB_III}. Despite so far the lack of a close formula for $\mathbf{E}(Y)$ that could allow possible connections to other well-known functions and problems, we can approximate $\mathbf{E}(Y)$ very easily and with as much accuracy as desired. We recall also that $H(\mathcal{C})$ from section \ref{SECT_gen_opt_algo} is equal to $\mathbf{E}(Y)$.

\begin{centering}
\renewcommand*{\arraystretch}{1.25}
\begin{longtable}{|l|l|l||l|l|l|}\caption{Expected height of $\mathcal{T}_{k}$ for $p=0.5$ and $p=0.51$}\label{TAB_I}\\
\hline
\multicolumn{3}{|c||}{$p=0.5$} & \multicolumn{3}{|c|}{$p=0.51$}\\
\hline
$k$ & $\mathbf{E}(Y_{k})$ & $\mathbf{E}(Y_{k})-\mathbf{E}(Y_{k-1})$ & $k$ & $\mathbf{E}(Y_{k})$ & $\mathbf{E}(Y_{k})-\mathbf{E}(Y_{k-1})$\\
\hline
\hline
\endfirsthead
\hline
\multicolumn{3}{|c||}{$p=0.5$} & \multicolumn{3}{|c|}{$p=0.51$}\\
\hline
$k$ & $\mathbf{E}(Y_{k})$ & $\mathbf{E}(Y_{k})-\mathbf{E}(Y_{k-1})$ & $k$ & $\mathbf{E}(Y_{k})$ & $\mathbf{E}(Y_{k})-\mathbf{E}(Y_{k-1})$\\
\hline
\hline
\endhead
\multicolumn{6}{r}{\boldgraytext{Continued on next page}}
\endfoot
\endlastfoot
\hline
$1$ & $1$ & \--\--\-- & $1$ & $0.9996$ & \--\--\-- \\
\hline
$2$ & $2.25$ & $1.25$ & $2$ & $2.249299$ & $1.249701$ \\
\hline
$3$ & $3.015625$ & $0.765625$ & $3$ & $3.016331$ & $0.767031$ \\
\hline
$4$ & $3.1016236$ & $0.085999$ & $4$ & $3.103287$ & $0.086957$ \\
\hline
$5$ & $3.1022065$ & $0.000583$ & $5$ & $3.1038984$ & $0.000611$ \\
\hline
$6$ & $3.1022065$ & $0.163457\cdot{}{10^{-7}}$ & $6$ & $3.1038984$ & $0.196847\cdot{}{10^{-7}}$ \\
\hline
$7$ & $3.1022065$ & $0.727524\cdot{}{10^{-17}}$ & $7$ & $3.1038984$ & $0.139177\cdot{}{10^{-16}}$ \\
\hline
$8$ & $3.1022065$ & $0.770550\cdot{}{10^{-36}}$ & $8$ & $3.1038984$ & $0.488859\cdot{}{10^{-35}}$ \\
\hline
$9$ & $3.1022065$ & $0.447530\cdot{}{10^{-74}}$ & $9$ & $3.1038984$ & $0.356004\cdot{}{10^{-72}}$ \\
\hline
$10$ & $3.1022065$ & $0.768362\cdot{}{10^{-151}}$ & $10$ & $3.1038984$ & $0.972355\cdot{}{10^{-147}}$ \\
\hline
\end{longtable}
\end{centering}

\begin{centering}
\renewcommand*{\arraystretch}{1.25}
\begin{longtable}{|l|l|l||l|l|l|}\caption{Expected height of $\mathcal{T}_{k}$ for $p=0.525$ and $p=1-e^{-1}$}\label{TAB_II}\\
\hline
\multicolumn{3}{|c||}{$p=0.525$} & \multicolumn{3}{|c|}{$p=1-e^{-1}$}\\
\hline
$k$ & $\mathbf{E}(Y_{k})$ & $\mathbf{E}(Y_{k})-\mathbf{E}(Y_{k-1})$ & $k$ & $\mathbf{E}(Y_{k})$ & $\mathbf{E}(Y_{k})-\mathbf{E}(Y_{k-1})$\\
\hline
\hline
\endfirsthead
\hline
\multicolumn{3}{|c||}{$p=0.525$} & \multicolumn{3}{|c|}{$p=1-e^{-1}$}\\
\hline
$k$ & $\mathbf{E}(Y_{k})$ & $\mathbf{E}(Y_{k})-\mathbf{E}(Y_{k-1})$ & $k$ & $\mathbf{E}(Y_{k})$ & $\mathbf{E}(Y_{k})-\mathbf{E}(Y_{k-1})$\\
\hline
\hline
\endhead
\multicolumn{6}{r}{\boldgraytext{Continued on next page}}
\endfoot
\endlastfoot
\hline
$1$ & $0.9975$ & \--\--\-- & $1$ & $0.930177$ & \--\--\-- \\
\hline
$2$ & $2.245622$ & $1.248122$ & $2$ & $2.125371$ & $1.195195$\\
\hline
$3$ & $3.020019$ & $0.774398$ & $3$ & $3.123057$ & $0.997685$\\
\hline
$4$ & $3.112037$ & $0.092017$ & $4$ & $3.405235$ & $0.282178$\\
\hline
$5$ & $3.112802$ & $0.000765$ & $5$ & $3.417840$ & $0.012606$\\
\hline
$6$ & $3.112802$ & $0.405387\cdot{}{10^{-7}}$ & $6$ & $3.417855$ & $0.149561\cdot{}{10^{-4}}$\\
\hline
$7$ & $3.112802$ & $0.827405\cdot{}{10^{-16}}$ & $7$ & $3.417855$ & $0.120218\cdot{}{10^{-10}}$\\
\hline
$8$ & $3.112802$ & $0.198626\cdot{}{10^{-33}}$ & $8$ & $3.417855$ & $0.417872\cdot{}{10^{-23}}$\\
\hline
$9$ & $3.112802$ & $0.594662\cdot{}{10^{-69}}$ & $9$ & $3.417855$ & $0.262288\cdot{}{10^{-48}}$\\
\hline
$10$ & $3.112802$ & $0.271310\cdot{}{10^{-140}}$ & $10$ & $3.417855$ & $0.526885\cdot{}{10^{-99}}$\\
\hline
\end{longtable}
\end{centering}

\begin{centering}
\renewcommand*{\arraystretch}{1.25}
\begin{longtable}{|l|l|l||l|l|l|}\caption{Expected height of $\mathcal{T}_{k}$ for $p=1/\sqrt{2}$ and $p=0.96875$ }\label{TAB_III}\\
\hline
\multicolumn{3}{|c||}{$p=1/\sqrt{2}$} & \multicolumn{3}{|c|}{$p=0.96875$}\\
\hline
$k$ & $\mathbf{E}(Y_{k})$ & $\mathbf{E}(Y_{k})-\mathbf{E}(Y_{k-1})$ & $k$ & $\mathbf{E}(Y_{k})$ & $\mathbf{E}(Y_{k})-\mathbf{E}(Y_{k-1})$\\
\hline
\hline
\endfirsthead
\hline
\multicolumn{3}{|c||}{$p=0.1/\sqrt{2}$} & \multicolumn{3}{|c|}{$p=0.96875$}\\
\hline
$k$ & $\mathbf{E}(Y_{k})$ & $\mathbf{E}(Y_{k})-\mathbf{E}(Y_{k-1})$ & $k$ & $\mathbf{E}(Y_{k})$ & $\mathbf{E}(Y_{k})-\mathbf{E}(Y_{k-1})$\\
\hline
\hline
\endhead
\multicolumn{6}{r}{\boldgraytext{Continued on next page}}
\endfoot
\endlastfoot
\hline
$1$ & $0.828427$ & \--\--\-- & $1$ & $0.121094$ & \--\--\--\\
\hline
$2$ & $1.935029$ & $1.106602$ & $2$ & $0.325676$ & $0.204582$\\
\hline
$3$ & $3.218082$ & $1.283053$ & $3$ & $1.032648$ & $0.706972$\\
\hline
$4$ & $3.888608$ & $0.670526$ & $4$ & $3.225589$ & $2.192941$\\
\hline
$5$ & $3.966054$ & $0.077446$ & $5$ & $9.000940$ & $5.775351$\\
\hline
$6$ & $3.966602$ & $ 0.000549$ & $6$ & $19.650858$ & $10.64992$\\
\hline
$7$ & $3.966602$ & $0.158247\cdot{}{10^{-7}}$ & $7$ & $29.516436$ & $9.865578$\\
\hline
$8$ & $3.966602$ & $0.715392\cdot{}{10^{-17}}$ & $8$ & $32.185252$ & $2.668815$\\
\hline
$9$ & $3.966602$ & $0.763973\cdot{}{10^{-36}}$ & $9$ & $32.270318$ & $0.085066$\\
\hline
$10$ & $3.966602$ & $0.445597\cdot{}{10^{-74}}$ & $10$ & $32.270366$ & $0.474475\cdot{}{10^{-4}}$\\
\hline
\multicolumn{3}{c|}{} & $11$ & $32.270366$ & $0.802734\cdot{}{10^{-11}}$ \\
\cline{4-6}
\multicolumn{3}{c|}{} & $12$ & $32.270366$ & $0.120147\cdot{}{10^{-24}}$ \\
\cline{4-6}
\end{longtable}
\end{centering}

All the entries from the previous tables were computed using the class for arbitrary-precision floating point
numbers of the NTL library from Shoup \url{https://libntl.org/}.

\section{Conclusion and future research}

The extraction algorithm constructed previously is iterative and is a generalization of the Von Neumann extraction algorithm. By modifying properly the tree of probabilities that one obtained by repeating the original Von Neumann's extraction algorithm, subsequences of different lengths from the biased source can be used to produce unbiased bits. From a programming point of view, the modifications pertains to prune the intermediate trees in such a way to minimize as much as possible the expected height of the resulting tree. The expected number of random biased bits required from the source was analyzed.

The work of Stout and Warren \cite{SouWar_1984} shows that there is no optimal algorithm with respect to the expected complexity. It seems however based on the references that no algorithm yields a better expected complexity than the iterative one constructed here, and the search for more efficient extraction algorithms continues. Is some knowledge about $p$ required to reach better expected complexity? Another line of research is when the source follows some Markovian processes or martingale processes, then the extraction becomes more complex, and the expected complexity more or less understood. Also the problem of transforming a sequence of non-uniform random combinatorial objects into another sequence of uniform, but different, or similar but with different properties, combinatorial objects has not been studied satisfactorily so far.

\bibliographystyle{ACM-Reference-Format}
\newcommand{\SortNoop}[1]{}

\end{document}